\documentclass[11pt]{article}
\usepackage{amsmath,amsthm,amssymb,mathtools,tikz,epsfig,enumerate}
\usepackage{hyperref,titling,titlesec,pdfpages,setspace,fancyhdr,multicol,appendix}
\usepackage[numbers]{natbib}
\usepackage[margin=1in]{geometry} 
\usepackage{pgfplots}
\usepackage{tkz-euclide}
\numberwithin{equation}{section}

\pgfplotsset{compat=newest}
\pgfplotsset{plot coordinates/math parser=false}
\usetikzlibrary{plotmarks}

\newlength\figheight
\newlength\figwidth
\newtheorem{thm}{Theorem}[section]

\theoremstyle{definition}

\newtheorem{eg}[thm]{Example}
 
\theoremstyle{remark}
\newtheorem{rem}[thm]{Remark}

\newcommand{\E}{\mathbb{E}} % expectation
\renewcommand{\P}{\mathbb{P}} % probability
\newcommand{\F}{\mathcal{F}} % filtration
\DeclarePairedDelimiter{\abs}{\lvert}{\rvert} % absolute value, or cardinality
\newcommand{\ts}{\textstyle}

\newcommand{\de}{\,\mathrm{d}}

\renewenvironment{thebibliography}[1]{%
\begin{oldthebibliography}{#1}%
\setlength{\baselineskip}{.8em}
\linespread{.8}
\small 
\setlength{\parskip}{0ex}%
\setlength{\itemsep}{.1em}%
}%
{%
\end{oldthebibliography}%
}

\begin{document}

\title{Optimal Trading with General Signals and \\ Liquidation in Target Zone Models}

\author{
Christoph Belak\thanks{University of Trier, Department IV -- Mathematics, Universit\"atsring 19, 54296 Trier, Germany, e-mail: \texttt{belak@uni-trier.de}.}
\and
Johannes Muhle-Karbe\thanks{Carnegie Mellon University, Department of Mathematical Sciences, 5000 Forbes Avenue, Pittsburgh, PA 15213, USA, email \texttt{ johannesmk@cmu.edu}. Parts of this research were completed while this author was visiting ETH Z\"urich. He thanks the Forschungsinstitut f\"ur Mathematik and H.~Mete Soner for their hospitality. }
\and
Kevin Ou\thanks{Carnegie Mellon University, Department of Mathematical Sciences, 5000 Forbes Avenue, Pittsburgh, PA 15213, USA, email \texttt{ yangxio@andrew.cmu.edu}.}
}

\date{August 1, 2018}

\maketitle

\begin{abstract}
We study optimal trading in an Almgren-Chriss model with running and terminal inventory costs and  general predictive signals about price changes. As a special case, this allows to treat optimal liquidation in ``target zone models'': asset prices with a reflecting boundary enforced by regulatory interventions. In this case, the optimal liquidation rate is the ``theta'' of a lookback option, leading to explicit formulas for Bachelier or Black-Scholes dynamics.
 



\end{abstract}

% introduction
\section{Introduction}

In classical models of optimal liquidation, the unaffected asset price is assumed to be a martingale \cite{bertsimas.lo.98,almgren2001optimal,obizhaeva.wang.13,alfonsi.al.10,predoiu.al.11}. This allows to focus on the liquidation program, while abstracting from signals about future price changes. The effect of such signals\footnote{Typical examples include order book imbalances or forecasts of the future order flow of other market participants.} is studied using stochastic control techniques in \cite{cartea.jaimungal.16,lehalle2017incorporating}, leading to a PDE characterization for Markovian signals and explicit formulas in the special case where the signal processes have Ornstein-Uhlenbeck dynamics. 

A rather different signal about future price changes is studied by~\cite{neuman.schied.16}. Motivated by caps on exchange rates enforced by central banks,\footnote{A recent example is the upper bound on the CHF/EUR exchange rate guaranteed by the Swiss National Bank.} they study the optimal liquidation of assets that reflect off an upper threshold.  Under the assumption that the liquidating agent only sells at this most favorable execution price,\footnote{Unlike in the standard optimal execution models surveyed above, the resulting optimal trading rate is singular, since it only acts on the local time of the reflected price process. Accordingly, the liquidity costs in the model of \cite{neuman.schied.16} is imposed on the trading rate in local time.} they characterize the optimal trading strategy by a PDE that admits a probabilistic representation in terms of catalytic superprocesses. 

Selling only at the highest possible price seems reasonable for agents with long liquidation horizons and low inventory costs. Yet, for shorter planning horizons or higher inventory costs, substantial immediate trading is necessary since it becomes too costly to wait for the asset price to approach its maximum.

In the present study, we solve the optimal liquidation problem with a price cap without constraining the selling price. To do so, we first extend the results of \cite{lehalle2017incorporating} on trading with running and terminal inventory costs to general, not necessarily Markovian trading signals by adapting the calculus-of-variations argument developed for optimal tracking problems in~\cite{bank2017hedging,bouchard2017equilibrium}. As a special case, this allows to treat optimal liquidation for the case of reflected price processes: these reduce to computing the ``theta''\footnote{That is, the derivative with respect to the time variable.} of a lookback call option. If the unaffected price process is modelled by a Bachelier or Black-Scholes model, the optimal trading rate can in turn be computed in closed form up to the numerical evaluation of an integral with explicit integrand. These results confirm the intuition outlined above. Indeed, we find that all sales occur close to the barrier if inventory costs are low. In contrast, for higher inventory costs, the influence of the barrier diminishes, as it becomes prohibitively expensive to hold an asset position while the asset price is far from its upper bound.

The remainder of this paper is organized as follows. The general model and the special case of a price cap are introduced in Section~\ref{s:model}. The solution of the general trading problem with inventory costs is subsequently derived in Section~\ref{s:result} and applied to optimal liquidation in models with a price cap in Section~\ref{s:cap}.

\paragraph{Notation}

Throughout, we fix a filtered probability space $(\Omega, \F, \{\F_t\}_{t\in[0,T]}, \P)$ satisfying the usual conditions and write $\E_t[\cdot] := \E[\cdot \vert \F_t]$ for $t\in[0,T]$. The set $\mathcal{H}^2$ denotes the special semimartingales whose canonical decomposition $S=M+A$ into a (local) martingale $M$ and a predictable finite-variation process $A$ satisfies $\E[\langle M\rangle_T] + \E[(\int_0^T \abs{\de A_t})^2] < \infty$. Finally, we write $\mathcal{V}$ for the progressively measurable processes $u$ satisfying $\E[\int_0^T \abs{u_t}^2 \de t] < \infty$.

\section{Model}\label{s:model}

We consider optimal trading in a risky asset with price process $P \in \mathcal{H}^2$. The asset position at time $t \in [0,T]$ is denoted by $X_t$, where the given initial position is $X_0:=x>0$. As in \cite{almgren2001optimal}, the position can only be adjusted gradually, since trades incur a cost $\lambda>0$ quadratic in the selling rate $u_t := -dX_t/dt$. With a running inventory cost $\gamma>0$ and a terminal inventory cost $\Gamma>0$, this leads to the following standard goal functional:\footnote{Note that $V(u)$ is well defined for all $u\in\mathcal V$. As the terminal inventory penalty $\Gamma$ grows, this criterion approaches optimal liquidation, where the position has to be closed out completely at maturity.}
\begin{equation}\label{eq:goal}
 V(u):= \E\Bigg[\underbrace{\int_0^T \overbrace{(P_t - \lambda u_t)}^{\textrm{execution price}} u_t \de t}_{\textrm{terminal cash position}} + \underbrace{P_T X_T}_{\stackrel{\textrm{terminal}}{\textrm{asset position}}} - \underbrace{\int_0^T \gamma X_t^2 \de t}_{\stackrel{\textrm{running}}{\textrm{inventory cost}}} - \underbrace{\Gamma X_T^2}_{\stackrel{\textrm{terminal}}{\textrm{inventory cost}}} \Bigg].
\end{equation}
For asset prices with martingale dynamics, criteria of this type were first introduced by \cite{almgren.12,forsyth.al.12} and subsequently studied by, e.g., \cite{schied.13,ankirchner.al.14,graewe.al.15}. If the asset price is of the form $\de P_t = I_t \de t + \de M_t$ for a one-dimensional diffusion $I$, then~\eqref{eq:goal} corresponds to the setup of \cite{lehalle2017incorporating}.  Here, we allow for more general -- potentially singular -- asset dynamics. This allows to cover the ``target zone models'' studied in \cite{neuman.schied.16}, where the asset price is capped at some finite level: 

\begin{eg}\label{ex:target}
Consider a martingale $M\in\mathcal{H}^2$ and a constant $\bar{P} \geq M_0$. Then the price ``capped'' at level $\bar{P}$ is defined as the solution of the Skorokhod map 
$$P:=M-(M^*-\bar{P})^+,$$
where $M^*_t:=\sup_{s\in[0,t]} M_s$. This corresponds to the minimal amount of intervention necessary to keep the asset price below level $\bar{P}$, akin to regulatory interventions to keep an exchange rate below a certain threshold.
\end{eg}

\section{General Solution}\label{s:result}

In our general - not necessarily Markovian - setup, the dynamic programming approach of \cite{lehalle2017incorporating} is no longer applicable. Instead, we adapt the calculus-of-variation argument of \cite{bank2017hedging,bouchard2017equilibrium} to the present setting, where the asset has a general -- possibly singular -- drift. 

As the goal functional $u\mapsto V(u)$ is strictly concave, it has a unique maximum $\hat{u}$ characterized by the first-order condition that the G\^ateaux derivative $V'(\hat{u})$ vanishes at this critical point. This leads to a system of linear forward-backward stochastic differential equations (FBSDEs) for the optimal trading rate and the position, which can be solved explicitly:

\begin{thm}\label{main}
Set 
$$\beta:=\sqrt{\gamma/\lambda}, \qquad G(t):= \beta\cosh(\beta t)+\lambda^{-1}\Gamma\sinh(\beta t).
$$
Let $P=P_0+M+A$ be the canonical decomposition of the (special) semimartingale $P$ into a local martingale $M$ and a finite-variation process $A$, and define 
\begin{align*}
v_2(t):= 
-\frac{G'(T-t)}{G(T-t)}, \quad
v_1(t):= \E_t\left[ \frac{1}{2\lambda}\int_t^T \frac{G(T-s)}{G(T-t)} \de A_s \right], \quad
v_0(t):= \E_t\left[ \int_t^T v_1(s)^2 \de s \right].
\end{align*}
The unique maximizer $\hat{u}$ of $u\mapsto V(u)$ over $\mathcal V$ solves the (random) linear differential equation
\begin{align}\label{eq:ODE}
&\hat{u}_t = -v_2(t) X^{\hat{u}}_t - v_1(t),
\end{align}
so that the optimal liquidation trajectory is given by
\begin{align}\label{eq:Pos}
&X^{\hat{u}}_t = \frac{G(T-t)}{G(T)}x + \int_0^t \frac{G(T-t)}{G(T-s)} v_1(s) \de s.
\end{align}
The corresponding optimal value for~\eqref{eq:goal} is
\begin{align}\label{eq:vf}
&\ts V(\hat{u}) = P_0 x + \lambda\left[v_0(0) + 2 v_1(0) x + v_2(0) x^2\right].
\end{align}
\end{thm}

\begin{proof}
We adapt the argument from \cite{bank2017hedging,bouchard2017equilibrium}. Recall that $X_t = x-\int_0^t u_s \de s$. Therefore, $X$ is an affine function of $u$. As the goal functional~\eqref{eq:goal} is a quadratic in $(u,X)$ with strictly negative quadratic coefficients, it admits a unique maximizer characterized by the critical point; see, e.g.,~\cite{ekeland1999convex}. We now solve for this critical point in feedback form.

\emph{Step 1: Compute the G\^ateaux derivative}. We fix a direction of variation $\alpha \in \mathcal{V}$ and compute
\begin{align*}
&\langle V'(u), \alpha \rangle = \lim_{\varepsilon \to 0}\tfrac{1}{\varepsilon}(V(u+\epsilon\alpha)-V(u))  \\
&= \E\left[\int_0^T P_t \alpha_t \de t - 2\lambda\int_0^T u_t \alpha_t \de t - P_T \int_0^T \alpha_t \de t + 2\gamma\int_0^T X_t \int_0^t \alpha_s \de s \de t + 2\Gamma X_T \int_0^T \alpha_t \de t \right] \\
&= \E\left[\int_0^T \alpha_t \E_t\left[P_t - 2\lambda u_t - P_T + 2\gamma\int_t^T X_s \de s + 2\Gamma X_T \right] \de t\right].
\end{align*}
This derivative has to vanish for any variation at the critical point $\hat{u}$, which is equivalent to
\begin{equation}\label{eq:foc}
\E_t\left[-2\lambda \hat{u}_t + P_t - P_T + 2\gamma\int_t^T X^{\hat{u}}_s \de s + 2\Gamma X^{\hat{u}}_T \right] = 0.
\end{equation}
We therefore obtain that the optimal trading rate $\hat{u}$ and the corresponding optimal position $X^{\hat{u}}$ solve the following system of linear forward-backward stochastic differential equations (FBSDE):
\begin{equation}\label{eq:FBSDE}
\begin{split}
\de X^{\hat{u}}_t &= - \hat{u}_t \de t, \quad X_0=x, \\
\de\hat{u}_t &= \frac{1}{2\lambda}\left(\de P_t-2\gamma X^{\hat{u}}_t \de t - \de N_t\right), \quad \hat{u}_T=\frac{\Gamma}{\lambda} X^{\hat{u}}_T.
\end{split}
\end{equation}
Here, $N$ is a square-integrable martingale that needs to be determined as part of the solution. 

\emph{Step 2: Solve the FBSDE for the critical point.} 
Setting,
\[
 Y:=\begin{pmatrix} X \\ u\end{pmatrix},\qquad Z:=\begin{pmatrix} 0 \\ P-N\end{pmatrix},\qquad B:=\begin{pmatrix} 0 & -1 \\ -\gamma/\lambda & 0\end{pmatrix},
\]
the FBSDE~\eqref{eq:FBSDE} can be written in vector form as 
\[
\de Y_t = B Y_t \de t + \frac{1}{2\lambda}\de Z_t, \quad Y^1_0=x, \quad (\Gamma/\lambda, -1)Y_T = 0.
\]
Integration by parts shows that $\de (e^{-B t}Y_t)=\frac{1}{2\lambda}e^{-B t}\de Z_t$ and in turn
\begin{align*}
Y_T &= e^{B(T-t)}Y_t+\frac{1}{2\lambda}\int_t^T e^{B(T-s)}\de Z_s.%= q(T-t)Y_t+\frac{1}{2\lambda}\int_t^T q(T-s) dZ_s.
\end{align*}
Now, multiply by $(\Gamma/\lambda,-1)$, take into account the terminal condition $(\Gamma/\lambda, -1)Y_T = 0$, and use  
$$
(\Gamma/\lambda, -1)e^{B(T-t)}=\begin{pmatrix} \frac{\Gamma}{\lambda}\cosh(\beta (T-t))+\beta \sinh(\beta(T-t)) \\ -\cosh(\beta(T-t))-\frac{\Gamma}{\lambda}\beta^{-1} \sinh(\beta(T-t))\end{pmatrix}=\frac{1}{\beta}\begin{pmatrix} G'(T-t) \\ -G(T-t) \end{pmatrix}.
$$
As a consequence,
$$
0=G'(T-t)X^{\hat{u}}_t-G(T-t)\hat{u}_t-\frac{1}{2\lambda}\int_t^T G'(T-s) \de (P_s-N_s).
$$
After solving for the trading rate and taking conditional expectations, this gives
$$
\hat{u}_t = -\frac{G'(T-t)}{G(T-t)} X^{\hat{u}}_t - \frac{1}{2\lambda} \E_t\left[\int_t^T\frac{G(T-s)}{G(T-t)} \de P_s\right].
$$
By the Doob-Meyer decomposition and since $P \in \mathcal{H}^2$, we can replace $\de P_s$ with $\de A_s$. The variation of constants formula now yields the explicit formula~\eqref{eq:Pos} for the corresponding optimal position $X^{\hat{u}}$. Since both $G$ and $G'$ are bounded from above and below away from zero, we have $\E[\abs{\hat{u}_t}^2] \le C_1 +C_2\int_0^t \E[\abs{\hat{u}_s}^2] \de s$ for some $C_1, C_2 > 0$. Gronwall's lemma in turn shows that $\E[\abs{\hat{u}_t}^2]\le C_1 e^{C_2 T}$ and hence $\hat{u}\in\mathcal V$ by Fubini's theorem.

\emph{Step 3: Compute the value function.}
The first-order condition $0=\langle V'(\hat{u}), \alpha \rangle$ for $\alpha=\hat{u}$ and its consequence~\eqref{eq:foc} for $t=0$ imply
\begin{align*}
V(\hat{u}) &= \frac{1}{2}\E\left[\int_0^T P_t \hat{u}_t \de t + P_T X^{\hat{u}}_T \right] + \frac{1}{2}x(-2\lambda \hat{u}_0 + P_0).
\end{align*}
Integration by parts as well as the Formulas~\eqref{eq:ODE} for $\hat{u}_0$ and \eqref{eq:Pos} for $X^{\hat{u}}_t$ in turn show that
\begin{align*}
V(\hat{u}) &=\frac{1}{2}\E\left[xP_0 + \int_0^T X^{\hat{u}}_t \de P_t \right] + \frac{1}{2}x(2\lambda v_2(0)x + 2\lambda v_1(0) + P_0)\\
&= xP_0 + \lambda(v_2(0)x^2+v_1(0)x) + \frac{1}{2}\E\left[\int_0^T \left(\frac{G(T-t)}{G(T)}x + \int_0^t \frac{G(T-t)}{G(T-s)}v_1(s) \de s\right) \de P_t\right].
\end{align*}
Since the price process $P \in \mathcal{H}^2$ can be replaced by its finite-variation part $A$ in the right-most expectation, the asserted form of the optimal value now follows from Fubini's theorem and the definition of $v_1(t)$. 
\end{proof}

\begin{rem}
The optimal trading rate~\eqref{eq:Pos} consists of two parts. The first prescribes to sell the initial position at a deterministic rate to reduce inventory. The second exploits signals about future price changes as summarized by a discounted conditional expectation of the asset's predictable drift. If the asset price $P$ is a martingale, the second term disappears and we obtain the classical optimal liquidation result of \cite{almgren2001optimal}. 
% In the setting of \cite{lehalle2017incorporating}, the signal process is absolutely continuous, $\de A_t = I_s \de t$. In the case of a general semimartingale asset price, our result simply confirms that the investor only needs to do the same thing as before to achieve the optimum: forecasting price trends until the end of the trading period. 
\end{rem}

\begin{rem}
The value function~\eqref{eq:vf} consists of two parts: $xP_0$ is the mark-to-market value of the position. The second part, $\lambda(v_0(0)+2v_1(0)x+v_2(0)x^2)$, is the expected risk-adjusted implementation shortfall of the meta order under the optimal strategy.
\end{rem}

% with price limit
\section{Solution for Target-Zone Models}\label{s:cap}

We now apply Theorem~\ref{main} to the target-zone model of Example~\ref{ex:target}, where
$$
P_t=M_t -(M^*_t-\bar{P})^+, \quad t \in [0,T],
$$
for a martingale $M\in\mathcal{H}^2$, its running maximum $M^*_t=\max_{s \in [0,t]}M_s$ and a constant price cap $\bar{P} \geq M_0$. The key to applying Theorem~\ref{main} is to compute the conditional expectation of future price changes. For the target zone models, we have
\begin{align*}
\E_t[P_s] = \E_t\left[M_s-(M^*_s-\bar{P})^+\right] &= M_t-\E_t\left[(M^*_s-\bar{P})^+\right]\\
& =M_t-(M^*_t-\bar{P})^+ - \E_t\left[(M^*_s-\bar{P})^+ - (M^*_t-\bar{P})^+\right]\\
&=P_t - \E_t\left[(M^*_s-\bar{P}\vee M^*_t)^+\right], \quad 0 \leq t \leq s \leq T.
\end{align*}
Thus, $\de A_s=\de_s \E_t[P_s] = -\de_s L_s(t)$, where $L_s(t):=\E_t[(M^*_s-\bar{P}\vee M^*_t)^+]$ is the price at time $t \in [0,T]$ of a lookback call option on $M$ with fixed strike $\bar{P}\vee M^*_t$ and maturity $s \in [t,T]$ . We have therefore reduced the computation of the optimal trading strategies from Theorem~\ref{main} to the calculation of the ``theta'' of a lookback call option written on the uncapped asset price. If the uncapped asset price follows arithmetic or geometric Brownian motions, the joint distribution of Brownian motion and its running maximum can in turn be used to compute the optimal trading strategy explicitly.

\subsection{Bachelier Model}
Suppose that $M_t:=M_0 + \sigma B_t$, where $B$ is a standard one-dimensional Brownian motion, $M_0\in\mathbb R$ and $ \sigma>0$ are constants, so that $L_s(t)$ is the price of a lookback call in the Bachelier model. A straightforward calculation shows that
\[
\de_s L_s(t) = \frac{\sigma}{\sqrt{s-t}}\phi\left(\frac{\bar{P}-P_t}{\sigma\sqrt{s-t}}\right)\de s,\qquad\text{where } \phi(x):=\frac{1}{\sqrt{2\pi}}e^{-\frac{1}{2}x^2}.
\]
Thus, the optimal trading rate from Theorem~\ref{main} is $\hat{u}_t = \bar{u}(t,X^{\hat{u}}_t,P_t)$, where
\begin{equation}\label{eqn:RateBachelier}
\bar{u}(t,x,p) := \frac{G'(T-t)}{G(T-t)} x +\frac{1}{2\lambda} \int_t^T \frac{\sigma}{\sqrt{s-t}} \frac{G(T-s)}{G(T-t)} \phi\left(\frac{\bar{P}-p}{\sigma\sqrt{s-t}}\right) \de s 
\end{equation}
and where the constant $\beta$ and the function $G$ are defined as in Theorem~\ref{main}. Setting
\[
 \bar{u}_{\mathrm{AC}}(t,x) := \frac{G'(T-t)}{G(T-t)} x\qquad\text{and}\qquad \bar{u}_{\mathrm{BA}}(t,p) := \bar{u}(t,0,p) = \bar{u}(t,x,p) -  \bar{u}_{\mathrm{AC}}(t,x),
\]
we observe that $\bar{u}_{\mathrm{AC}}$ is the optimal trading speed in the absence of a price cap, cf.~\cite{almgren2001optimal}, which does not depend on the current asset price. In contrast, with a price cap, the optimal trading speed also depends on the distance of the current asset price $P_t$ from the cap $\bar{P}$ through $\bar{u}_{\mathrm{BA}}$. Since $\bar{u}_{\mathrm{BA}}\ge 0$, the position is liquidated at a higher rate if a price cap is present.  

\begin{figure}[ht]
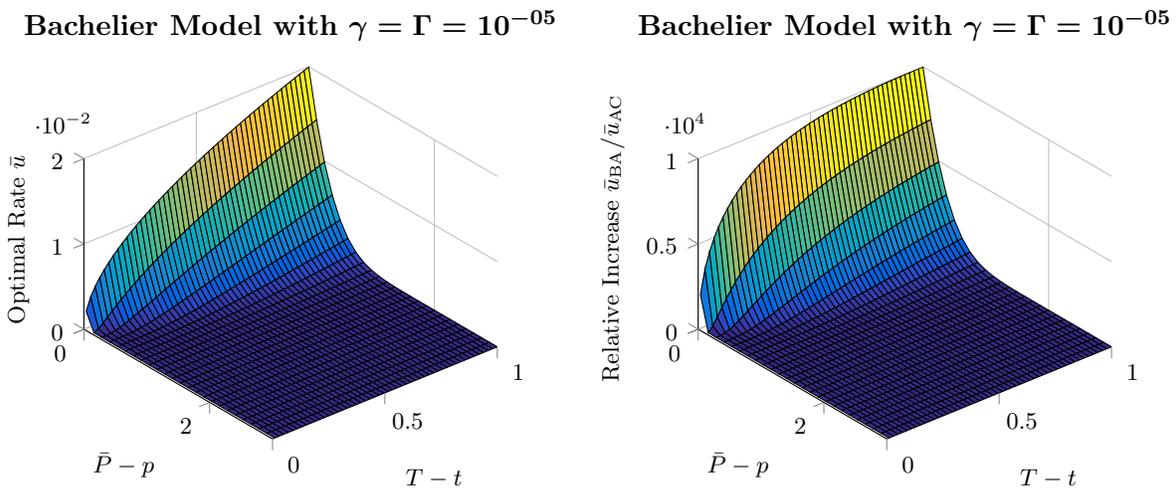

 \setlength\figheight{0.3\linewidth}
 \setlength\figwidth{0.35\linewidth}
 \centering
 \input{figABM_1e-05.tikz}\hspace{0.5em}
 \input{figDiffABM_1e-05.tikz}
 \caption{Optimal trading rate (left) and relative increase of the trading rate over the Almgren-Chriss rate (right) in the Bachelier model in the low inventory regime $\gamma = \Gamma = 10^{-5}$. The other model parameters are  $T = 1$, $x = 1$, $\lambda = 0.1$, $\sigma = 0.5$.}
 \label{fig:Bachelier:low}
\end{figure}

The additional rate $\bar{u}_{\mathrm{BA}}(t,p)$ is decreasing in $\bar P-p$ and maximized at $p = \bar P$, i.e., most additional trading happens when the current asset price is near the price cap. In \cite{neuman.schied.16}, it is assumed that the asset is only sold when its price $p$ coincides with the price cap $\bar P$. Our unconstrained solution shows that this assumption is justified for small inventory costs $\gamma,\Gamma$. Indeed, suppose for simplicity that $\gamma = \Gamma$, so that the function $G$ depends on the model parameters only through $\beta = \sqrt{\gamma/\lambda}$:
\[
 G(T-t) = G(T-t;\beta) = \beta\cosh(\beta(T-t)) + \beta^2 \sinh(\beta(T-t)).
\]
One then immediately verifies that
\[
 \lim_{\beta\downarrow 0} \bar u_{\mathrm{AC}}(t,x;\beta) = 0\qquad\text{and}\qquad \lim_{\beta\downarrow 0} \bar u_{\mathrm{BA}}(t,p;\beta) = \frac{1}{2\lambda} \int_t^T \frac{\sigma}{\sqrt{s-t}} \phi\left(\frac{\bar{P}-p}{\sigma\sqrt{s-t}} \right) \de s.
\]
Whence, for small inventory costs $\gamma=\Gamma$, the optimal trading rate is largely determined by $\bar{u}_{\mathrm{BA}}$. In particular, the trader sells at a high rate if the asset price $p$ is close to the price cap $\bar P$, whereas the trading rate vanishes as the difference $\bar P - p$ becomes large. For  $\gamma = \Gamma = 10^{-5}$ and 
\[
 T = 1,\qquad x = 1,\qquad \lambda = 0.1,\qquad \sigma = 0.5,
\]
this is illustrated in Figure~\ref{fig:Bachelier:low}. There, we plot the optimal trading rate $\bar u$ and the relative increase $(\bar u - \bar u_{\mathrm{AC}}) / \bar u_{\mathrm{AC}}= \bar u_{\mathrm{BA}} / \bar u_{\mathrm{AC}}$ of the optimal rate compared to the Almgren-Chriss solution as functions of length of the liquidation period $T-t$ and moneyness $\bar P - p$.

We observe that the optimal trading rate in this case is almost equal to zero if the asset price $p$ is away from the price cap $\bar P$, i.e., the asset is only sold near the price cap. In particular, at this critical level $\bar P$, the trading rate with price cap is up to $10^4$ times higher than the trading rate without price cap. Therefore, in this low inventory cost regime, only allowing trades at $p = \bar P$ as in \cite{neuman.schied.16} is a reasonable approximation.

The corresponding results for higher inventory costs $\gamma=\Gamma=1$ are reported in Figure~\ref{fig:Bachelier:moderate}. (All the other model parameters are the same as for the previous example.) We observe that the optimal rate is decreasing as a function of moneyness $\bar P-p$ and increasing as a function of the liquidation period $T-t$. The plot of the relative increase $\bar u_{\mathrm{BA}} / \bar u_{\mathrm{AC}}$ shows that as the asset price $p$ approaches the price cap $\bar P$, there is an increase in the liquidation rate of initially more than 30\%. This additional effect is more pronounced if the trading horizon $T-t$ is large and vanishes for small liquidation periods. However, except for small values of moneyness, the qualitative shape of the optimal rate is for the most part determined by the Almgren-Chriss rate $\bar u_{\mathrm{AC}}$. In other words, unless the asset price is close to the price cap, the trader with higher inventory costs essentially neglects its presence.

\begin{figure}[ht]
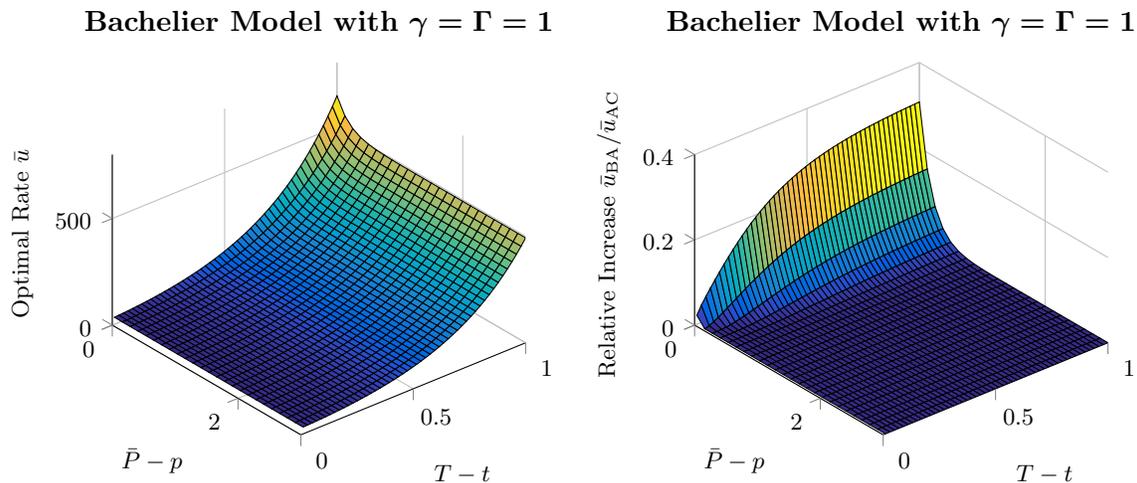

 \setlength\figheight{0.3\linewidth}
 \setlength\figwidth{0.35\linewidth}
 \centering
 \input{figABM_1.tikz}\hspace{0.5em}
 \input{figDiffABM_1.tikz}
 \caption{Optimal trading rate (left) and relative increase of the trading rate over the Almgren-Chriss rate (right) in the Bachelier model in the moderate inventory regime $\gamma = \Gamma = 1$. The other model parameters are  $T = 1$, $x = 1$, $\lambda = 0.1$, $\sigma = 0.5$.}
 \label{fig:Bachelier:moderate}
\end{figure}

\subsection{Black-Scholes Model}
Now suppose that the uncapped asset price $M$ follows a geometric Brownian motion, $M_t:=M_0\exp\left(\sigma B_t-\frac{1}{2}\sigma^2 t\right)$ for a standard Brownian motion $B$ and constants $S_0, \sigma >0$. Then, the computation of the optimal selling rate boils down to the computation of the theta of a lookback call in the Black-Scholes model. A standard calculation shows
\[
\de_s L_s(t) = M_t\left[ \frac{\sigma}{\sqrt{s-t}}\phi\bigl(f(s-t,M_t,P_t)\bigr)+\frac{\sigma^2}{2}\Phi\bigl(f(s-t,M_t,P_t)\bigr)\right] \de s,
\]
where $\Phi$ denotes the cumulative distribution function of the standard normal law, and
\[
 f(u,m,p) := \frac{\sigma\sqrt{u}}{2} - \frac{1}{\sigma\sqrt{u}}\log\left(\frac{\bar{P}-p}{m}+1\right).
\]
Therefore, the optimal liquidation rate $\hat{u}$ is 
\[
 \hat{u}_t = \bar{u}_{\mathrm{AC}}(t,X^{\hat{u}}_t) + \bar{u}_{\mathrm{BS}}(t,M_t,P_t),
\]
where $\bar{u}_{\mathrm{AC}}$ is defined as in the Bachelier model and
\[
 \bar{u}_{\mathrm{BS}}(t,m,p) := \frac{m}{2\lambda}\int_t^T  \frac{G(T-s)}{G(T-t)} \left[ \frac{\sigma}{\sqrt{s-t}}\phi\left(f(s-t,m,p)\right)+\frac{1}{2}\sigma^2 \Phi\bigl(f(s-t,m,p)\bigr)\right] \de s.
\]
As in the Bachelier model, the trader liquidates the position at a higher rate if a price cap is present since $\bar{u}_{\mathrm{BS}}\ge 0$, and the effect is more pronounced for small trading costs $\lambda$. Moreover, we have
\[
 \frac{\de }{\de f} \Bigl[\frac{\sigma}{\sqrt{s-t}}\phi\left(f\right)+\frac{1}{2}\sigma^2 \Phi\bigl(f\bigr)\Bigr] = \frac{\sigma}{\sqrt{s-t}}\phi(f)\Bigl[\frac{\sigma\sqrt{s-t}}{2} - f\Bigr],
\]
from which we infer that $\bar{u}_{\mathrm{BS}}$ is decreasing in $\bar P - p$ since $f(s-t,m,p) \leq \sigma\sqrt{s-t}/2$ and $f$ is clearly decreasing in $\bar P - p$. This is again in line with our findings in the Bachelier model. In contrast to the Bachelier model, however, the additional rate $\bar{u}_{\mathrm{BS}}$ also depends on the uncapped asset price $M_t$ and the dependence is monotonically increasing. Nevertheless, the numerical results in the Black-Scholes model are qualitatively very similar to their counterparts for the Bachelier model, compare~\cite{forsyth.al.12}; we therefore do not report them here.

\bibliographystyle{plainnat}
\phantomsection
\addcontentsline{toc}{section}{\refname}
\bibliography{References}

\end{document}